\documentclass[11pt,onecolumn,draftclsnofoot,letterpaper]{IEEEtran}
\usepackage{graphicx}
\usepackage{subfigure}
\usepackage{epsfig}
\usepackage{amssymb}
\usepackage{amsmath}
\usepackage{amsthm}
\usepackage{amsfonts} 
\usepackage{multirow}
\usepackage{pifont}
\usepackage{geometry}
\newcommand {\comment}[1]{}
\newcommand{\vc}{\mbox{${\bf c}$}}

\newcommand{\vx}{\mbox{${\bf x}$}}

\newcommand{\mA}{\hbox{{\bf A}}}

\newcommand{\mC}{\hbox{{\bf C}}}

\newcommand{\mE}{\hbox{{\bf E}}}

\newcommand{\mU}{\mbox{{$\bf U$}}}

\newcommand{\ga}{\alpha}

\def\bm#1{\mbox{\boldmath $#1$}}
\newcommand{\vga}{\mbox{$\bm \alpha$}}

\newcommand{\vgd}{\mbox{$\bm \delta$}}

\newcommand{\vgl}{\mbox{$\bm \lambda$}}

\newtheorem{theorem}{Theorem}[section]
\newtheorem{lemma}[theorem]{Lemma}

\newtheorem{prop}{Proposition}[section]
\newtheorem{claim}{Claim}[section]

\newtheorem{definition}{Definition}[section]
\newtheorem{question}{Question}[section]
\newtheorem{coro}{Corollary}[section]

\newcommand{\beq}{\begin{equation}}
\newcommand{\eeq}{\end{equation}}
\newcommand{\bea}{\begin{array}}
\newcommand{\ena}{\end{array}}
\newcommand{\bds}{\begin {itemize}}
\newcommand{\eds}{\end {itemize}}
\newcommand{\bdf}{\begin{definition}}
\newcommand{\blm}{\begin{lemma}}
\newcommand{\edf}{\end{definition}}
\newcommand{\elm}{\end{lemma}}
\newcommand{\bthm}{\begin{theorem}}
\newcommand{\ethm}{\end{theorem}}
\newcommand{\bprp}{\begin{prop}}
\newcommand{\eprp}{\end{prop}}
\newcommand{\bcl}{\begin{claim}}
\newcommand{\ecl}{\end{claim}}
\newcommand{\bcr}{\begin{coro}}
\newcommand{\ecr}{\end{coro}}
\newcommand{\bquest}{\begin{question}}
\newcommand{\equest}{\end{question}}

\title{On the allocation of  multiple divisible assets to players with different utilities}
\author{Ephraim Zehavi$^1$, Amir Leshem$^1$
\thanks{The authors are with the $^1$Faculty of Engineering, Bar-Ilan
University, Ramat-Gan, 52900, Israel.  e-mail: ephraim.zehavi@biu.ac.il. }}

\date{\today}
\begin{document}
\bibliographystyle{ieeetr}
\maketitle

\begin{abstract}
When there is  a dispute between players on how to divide multiple divisible assets, how should it be resolved? In this paper we introduce  a multi-asset game model that enables cooperation between multiple agents who bargain on sharing   $K$ assets,  when each player has a different value for each asset. It thus extends the sequential discrete  Raiffa solution  and the Talmud rule  solution to  multi-asset cases.

keyword: resource allocation, game theory, Raiffa Bargaining Solution, Aumann Bankruptcy, non-transferable commodities
\end{abstract}

%-----------------------------------------------------------------------------------------------------------------------------------------------------------
%  I. Introduction
%-----------------------------------------------------------------------------------------------------------------------------------------------------------
\section{Introduction}
The study of bargaining between players who can benefit from cooperation dates  back to the beginnings of game theory \cite{nash50}.  Over  the years many different solutions to the bargaining problem have been proposed. A good  overview of  bargaining solutions and models  can be found in a volume by  \cite{osborne90} and the references therein.
In contrast the age-old problem of adjudication of conflicting claims has been dealt with by all societies probably since the dawn of civilization.  The contemporary study of the mathematical problem of resolving conflicting claims can be attributed to  \cite{oneill1982}, where he formulated the problem in a game theoretic manner and solved it using cooperative game theoretic techniques. In the last  thirty  years  there has been  extensive exploration of  the axiomatic bases of bargaining solutions and ways to resolve conflicting claims in bankruptcy  cases. An excellent recent overview can be found  in \cite{thomson2013, thomson2013a},  which extends Thomson's overview of older results  on the relationship between bargaining and the adjudication of conflicting claims \cite{thomson_2003,thomson_2012}.

 There are several alternative approaches to analyzing collaborative solutions. One approach is based on building an axiomatic structure that leads to  a single solution. This approach began with the Nash bargaining solution \cite{nash50,nash53} and includes the analysis of many other solutions;
e.g., the Raiffa solution \cite{raiffa1953}, the Kalai-Smorodinsky solution, \cite{kalai75} and the family of generalized Nash solutions, \cite{kalai77b}.
Other approaches emphasize the negotiation  process to reach a final agreement. \cite{salonen_1988} was the first to establish a step-by-step axiomatic definition of  the discrete Raiffa solution for the $N$-player bargaining problem, based on  four axioms.  \cite{livne_1989}, as well as  \cite{peters_1991} presented  characterizations of the  continuous Raiffa  solution. Recently,   \cite{trockel_2009} suggested viewing the discrete Raiffa  solution as a repetition of a process  based on  three standard axioms; namely (a) Pareto optimality (b) invariance to affine transformation, and (c) symmetry \cite{diskin_2011}  generalized the Raiffa solution to the case of multi players achieving interim settlements step-by-step. They  defined a family of discrete solutions for N-person bargaining problems  which approaches the continuous Raiffa solution as the step size gradually becomes smaller.  \cite{anbarc_2013} proposed a unified framework for characterizations of different axioms that lead to different bargaining solutions. Their solution was simplified by \cite{trockel2014a} who also filled in a gap in the proof. Recently, Trockel also proposed an alternative formulation for the discrete Raiffa solution based on non-transferable utility games. \cite{trockel2014b}
Another approach  is to define  a bargaining process that leads to a specific bargaining solution. \cite{myerson_1991},    \cite{tanimura_2008} and    \cite{trockel_2011}  proposed  a mechanism for reaching bargaining solutions in which two players are allowed to make a sequence of simultaneous  propositions and to converge to the discrete Raiffa solution.

The \cite{aumann_1985}  bankruptcy solution is based on an interpretation of two claim resolution scenarios  discussed in the Talmud. The first case  is the Contested Garment (CG) problem where two men disagree on the ownership of a garment\footnote{ Mishna Baba Metzia 2a: The first man claimed half of it belongs to him and the other claimed it all; the decision was that  the one who claimed half is awarded $1/4$ and the other is awarded $3/4$.  The principle is clear: the first man agrees  that half of the  garment does  not belong to him. Therefore, the  bargaining is only on half of the garment.}. The second case\footnote{Kethubot 93a: a man married three women. The first woman had a marriage contract of $100$, the second of $200$, and the third of $300$. The man dies and his estate is worth $E$.  The ruling of Rabbi Nathan was as follows: If the estate  is worth $E=100$, then the estate will divided equally, namely  $33\frac{1}{3}$ for each. If the estate is worth $200$ the division will be $(50,75,75)$ and if  it is worth $E=300$ the division is  $(50,100,150)$, respectively.} addresses the estate division problem among  three women.  \cite{aumann_1985} constructed two rules that generalize the Talmud rules and can be applied to resolve the bankruptcy problem. Later these rules were generalized by Moreno-Ternero and Villar who defined a family of rules termed TAL,  \cite{moreno2006}. \cite{thomson2008} extended the Talmud rules even further by considering a wider family of rules which he termed ICI and CIC.
Another line of research was taken by  \cite{dagan_1993}. They represented bankruptcy problems as bargaining problems. This enabled them to study the Nash bargaining and the Kalai Smorodinsky solutions as means of solving the bankruptcy problems. Specifically, they proved that the Nash bargaining solution induces the constrained equal award rule and the Kalai-Smorodinski induces the proportional rule.
In the other direction, \cite{quant2006} suggested that certain games; i.e., the class of compromise admissible games with transferable utility, can be considered as coalitional games  and  their solution was related to the run-to-the-bank (RTB) rule, showing that in certain cases bankruptcy solutions can be the basis for the solution to cooperative games. This leads directly to the query of whether the Raiffa solution can be described as an iterative application of a bankruptcy problem. This is indeed one of the goals of this paper.

An interesting generalization of the bankruptcy problem to multiple issues was first proposed by  \cite{calleja2005}. In their formulation each agent has multiple claims regarding the total assets, and each claim is related to a different issue. They defined a division problem with multiple issues as follows:
\bdf
\label{def1}
Let the set of agents be $I=\left\{1,...,N\right\}$. Each agent has a vector of claims $\vc_n=(c_{n1},...,c_{nK})$, regarding the issues $1,...,K$ and for all $n,k$ $c_{nk}\ge0$. A multiple issue
bankruptcy problem is given by the pair $\left(C,E\right)$ where $C=\left[\vc_1,...,\vc_N\right]^T$,  $E>0$ is the total value of  the assets that should be divided among the agents and
$c=\sum_{n=1}^N \sum_{k=1}^K c_{nk}>E$.
A vector $\vx=\left(x_1,...,x_N \right)$ is efficient if $\sum_{n=1}^N x_n=E$.
A division rule for a multi-issue claim problem is a function that assigns to each multi-dimensional claim problem $\left(C,E\right)$ an efficient vector $\vx=f\left(C,E\right)$.
\edf
%------------------- Till Here
They proposed a multi-dimensional extension of the run-to-the-bank rule in \cite{oneill1982} and showed that it coincides with the Shapley value for the generated coalitional game.
Based on the work of Calleja et al,  \cite{gonzalez2007} showed that the multi-dimensional run-to-the-bank rule always yields a core element, and that it satisfies self-duality.
\cite{ju2007} provides a good discussion of the different problems that can be represented as a multiple issue bankruptcy problem.
An extended discussion of the problem of multi-issue bankruptcy is presented in \cite{hinojosa2012} where it is shown that the theory of cooperative games provides an allocation rule consistent with the Talmud rule \cite{aumann_1985} in the case of two agents.  It is worth noting that the literature on multiple issue claim problems has only been concerned with allocating vectors on the face  of the $N$ dimensional simplex (the face containing the efficient allocation vectors). However, this does not cover the most general case of claim problems.
 When there are multiple assets, these are replaced by the total worth of all the assets. In many bankruptcy problems this is indeed desirable for operational simplicity of the bankruptcy problem. However, different types of assets can definitely have different value for  different agents. For example, the agents may be subject to different taxation laws, in which case they might prefer to have larger share of one type of asset  rather than others (for example, the taxation of property, equity in companies and cash differ significantly depending on the countries. In this sense, we assume that utility is not transferable between agents. This generalization  is the main focus of this paper.
In this case the differences between bankruptcy and general bargaining tend to decrease, since in both cases different agents have different utilities for each division of the assets. The non-transferable utility claim problem can now be formulated as follows:
 \bdf
 Assume that we have $N$ agents and a vector of assets $\mE=\left(E_1,...,E_K\right)$. We assume that each of the assets is divisible such that any agent can get a part of each asset.
 Agent $n, \ \  1 \le n \le N$ has claims $c_{n1},...,c_{nk}$ for each of the assets. Furthermore, each agent has a utility associated with each asset $u_{nk}$. A generalized claim problem is given by a triple $\left(\mE,\mC,\mU\right)$. The allocation matrix $\mA=(\ga_{nk}).$ The total utility for player $n$ is given by $\sum_{k=1}^K \ga_{nk} u_{nk}$. The allocation rule for a generalized claim problem is a function $f\left(\mE,\mC,\mU\right)=\mA$ where $\mA$ is a stochastic allocation matrix, i.e., $\sum_{n=1}^N \ga_{nk}=1$ for every $k$.
 \edf

This generalizes  Haake's model, \cite{haake_2009}, who studied the Perles-Maschler solution and the discrete Raiffa solution for two agents sharing $K$ divisible commodities with utilities that are linear in the share in each commodity.  Haake derived a procedure based on pricing. In an article on  frequency allocation problems \cite{leshem_2006,leshem_2008} discussed a more general resource allocation problem where the utilities are convex functions of the  weighted linear sum of the assets (a slightly more general model  than the generalized claim problem). 
 They provided an efficient algorithm for computing the NBS and showed that in the two agents case only a single commodity is shared and the computational complexity is $K\log (K)$.
 They also discussed the NBS in the general case of $N$ players and $K\geq 2$ commodities  \cite{leshem_2008} and showed that there is always a solution where at most ${N}\choose{2}$ commodities are shared and any other commodity is allocated to one of the agents. This was extended and similar results were shown to hold for the Generalized Nash solutions and the Kalai-Smorodinsky solutions (KSS) in  \cite{zehavi_2009}.
A related problem was  considered by  \cite{ponsati_1997} and by \cite{marmol_2007}  who  addressed  the problem of resolving global bargaining problems over a finite number of different issues. They defined  max-min and leximin global bargaining solutions. However, their rules allowed the agents to back away from agreements on previous issues if required. This was done through a comprehensive extension of the sum of the previously agreed points and the new point.
Interestingly, under quite general conditions, a recent result reported in  \cite{zehavi_2013} showed that the optimal solution for any {P}areto optimal solution can always be achieved by allocating all the utility related to each issue to one of the agents in all but $N-1$ issues. By extending the analysis of the pareto boundary by \cite{mjelde83} they showed that this is also valid for the (generalized) NBS and the KS solutions. Since under any issue-by-issue negotiation it is unlikely that for any given issue all the value related to an issue will be allocated to a single agent, issue-by-issue bargaining with comprehensive extension is limited. In fact  the comprehensive extension \cite{marmol_2007} conceals a renegotiation of agreements in the previous stages.

In this paper  we  extend both  the discrete Raiffa  bargaining  solution \cite{raiffa1953} and the Talmud rule (TR)  \cite{aumann_1985} for resolving the allocation of $K$ assets to $N$ agents when the utility of each player is a convex function of the linear sum of its utility for each asset.  In the discrete  Raiffa solution the players reach an agreement  step by step on an intermediate partition of the utility. However, if some utility is left over, all the players continue to solve the problem until Pareto optimality is achieved.  The Talmud rule  bankruptcy solution is based on an extension of a Talmudic approach involving two  individuals claiming a single garment to resolve a dispute between heirs.
The structure of the paper is as follows: In section \ref{sec:The model} we describe the model of the  bargaining game and define a unified notation. In section \ref{sec:section3} we discuss the generalizations of the Raiffa and the Talmud rule to $K$ assets. In section \ref{sec:section4} we discuss the properties of the solutions and provide some detailed examples. We conclude that global bargaining solutions can be obtained by solving a sequence of linear programming problems.

%-----------------------------------------------------------------------------------------------------------------------------------------------------------
%  II section: The model and definitions
%-----------------------------------------------------------------------------------------------------------------------------------------------------------
\section{The model and definitions}
%=============================================================================
\label{sec:The model}

%========================================================
% Game definition
%=========================================

An N-player multi assets bargaining game is described by the set of players ${\cal{N}}=\{1,...,N\}$, where each player has  access to ${\cal{K}}=\{1,...,K\}$ assets.  The utility of the $k$'th asset  to the $n$'th  player is $u_{nk}>0$, and the utility functions are additively separable across assets. The utility vector of player $n$ is $\bm{u_n}=(u_{n1},\cdots,u_{nK})$. The claims of the players are a vector $\bm{u}=(\|\bm{u_1}\|_1,\cdots), \|\bm{u_N}\|_1)$, where $\|(*)\|$ is $L^1$ norm of $(*)$.  The action of player $n$ is a vector ${\bm{\alpha_n}}=\left(\alpha_{n1}, \cdots,\alpha_{nK}\right)$, where $0\leq\alpha_{nk}\leq1$ and  $\|\bm{\alpha_n}\|_1 =1$.  Let ${\bm{d}}=(d_1,  \cdots, d_N)^T\in  \mathbb{R}^{N}$, be the  disagreement vector  of the players, where $d_n$ is the  disagreement utility  of the $n$'th player. 
The bargaining game problem is given by a triple $(\bm{S},\bm{d},\bm{u})$, where $\bm{S} \subset \mathbb{R}^{NK}$  is a compact, convex, and comprehensive set  of all possible results of the allocation, where any allocation matrix $\bm{A}=(\bm{\alpha_1},\cdots,\bm{\alpha_N})^T$ induces a point in $\bm{S}$. A bargaining game solution is a function  $\varphi(\bm{S,d,u})=A \in \mathbb{R}^{NR}$ that uniquely defines the allocation matrix. We denote by $\varphi_n(\bm{S,d,u})=\bm{\alpha_n} \in \mathbb{R}^{K}$ the allocation vector to player $n$.   The point $\bm{s}=(\bm{s_1},\cdots,\bm{s_n})^T\in\bm{S}$ is defined by the allocation matrix, where   $\bm{s_n}=(\alpha_{n1}u_{n1},\cdots,\alpha_{nK}u_{nK})$. The sum of the utilities allocated to player $n$  is $s_n={\bm{\alpha_n}}\bm{u_n}^T$.
The intended interpretation of the actions of the players  is as follows: the utility matrix $\bm{s}$ is the von Neumann-Morgenstern utility level attained by the  players through the choice of some joint action. The players can achieve $\bm{s}$ if they unanimously agree on an allocation matrix $\bm{A}$. If they do not agree on any point in $(\bm{S},\bm{d})$ they end up at $\bm{d}$. 
%==============================================================================
%Def: The SuM Utility Space 
%==============================================================================
 \bdf
Any allocation matrix $\bm{A}$ induces a vector of total utilities  $\bm{v}=(s_1,\cdots,s_N)\in \mathbb{R}^N$. There are many to one mapping from the  set $\bm{S}\subset \mathbb{R}^{NK}$ to  the set $\bm{V}\subset \mathbb{R}^{N}$  of all feasible vectors $\bm{v}$. 
 \edf
%==============================================================================
%Def: individually rational
%==============================================================================
 \bdf
 The individually rational part of \bm{S} is all the points $\bm{s}\in\bm{S}$ that provide a higher utility than the disagreement utility to all players, \\
 i.e $ \bm{S}_d=\left\{\bm{s}|\|\bm{s_n}\|_1\geq d_n, \forall n,  \bm{s}\in\bm{S}\right\}$.
 \edf

 Players agree to negotiate only if they can get more than their disagreement  point. Note  that in  a multi-asset game the interest of a player $n$ is to maximize the sum of the utilities, i.e. $\|s_n\|_1$. The way in which the assets are combined in the allocation makes no difference to the players as long as the selection results in maximum total utility. Solutions that allocate the same utility to each player are equivalent solutions. The objective of the game is to find the point on the {P}areto frontier of set $\bm{S_d}$. More formally we use the following definitions:
 
 %==============================================================================
%Def: Pareto efficient
%==============================================================================
\bdf
Let $\bm{S}\subset \mathbb{R}^{NK}$ be a set. Then $\bm{s}\in \bm{S}$ is {P}areto
efficient if there is no $\bm{x}\in \bm{S}$ for which $\bm{x_n}>\bm{s_n}$ for all $n \in N$;  $\bm{s}\in \bm{S}$
is strongly {P}areto efficient if there is no $\bm{x}\in \bm{S}$ for
which $\bm{x_i}\geq \bm{s_i}$ $\forall i$, and $\bm{x_i}> \bm{s_i}$ for some $i \in N$. The
{P}areto frontier is defined as the set of all $\bm{s}\in \bm{S}$ that are
{P}areto efficient,  and is denoted by $\partial \bm{S}$.
\edf
%=============================================f=================================
%Def: $\epsilon$-Pareto efficient
%==============================================================================
\bdf
Let $\bm{S}\subset \mathbb{R}^{NK}$ be a set,  and  $\partial \bm{S}$ is the {P}areto frontier of the set. Then, $\bm{x}\in \bm{S}$ is $\epsilon$-{P}areto
efficient if there is  $\bm{s} \in \partial \bm{S}$ for which $\vert \bm{x_i}-\bm{s_i}\vert <\epsilon$  for all $i \in N$.
\edf

 %==============================================================================
%Def:Restricted  Ideal Point 
%==============================================================================
 
\bdf
{Restricted  Ideal point  (RIP)  rule}: For every bargaining  problem $(\bm{S, d,u})$, there  is an  ideal allocation  for player $n$, $\bm{A_n}=\{\bm{\alpha_p^{(n)}}\}_{p=1}^N$.  $\bm{A_n}$  is the  allocation matrix that maximizes the total utility of the  player $n$,  while maintaining the utility  $d_{p}$ for all other players $p\neq n$. $\bm{A_n}$ is located on the {P}areto frontier of  set $\bm{S}$.
The ideal allocation allocated to player $n$ is the  total utility $I_n(\bm{S, d,A_n})=s_n^{(n)}$. We denote by $I_n(\bm{S, d,A_n})=s_n^{(n)}$  the  ideal point of player $n$.  $\bm{A_n}$ is a solution of the following  linear programming problem:
\beq
\begin{array}{rl}
\label{RIPeq}
\textbf{max}      &  \bm{\alpha_n^{(n)}}\bm{u_n}^T \\
\textbf{subject to:} & \|\alpha_ {k}^{(n)}\|_1 = 1, \forall k, \quad \\
		    &\\
                     &  \alpha_ {nk}^{(n)} \geq0,  \forall n,k,\\
                     &\\
                     &  \bm{\alpha_{p}^{(n)}}\bm{u_p}^T= d_p, \forall p\neq n.,\\
                     &\\
                     &  \bm{\alpha_{n}^{(n)}}\bm{u_n}^T\geq d_n. 
\end{array}
\eeq
\edf
The vector  $I(\bm{S, d,A})=(I_1(\bm{S, d,A_1}),\cdot,I_N(\bm{S, d,A_n})\in\mathbb{R}^N$. 
 %==============================================================================
%Def:  midpoint
%==============================================================================

 \bdf{ The Midpoint  (MP) Rule:}

 The midpoint is  the mapping  that maps the set of ideal  points $\{I_n(\bm{S, d,A_n})\}_{n=1}^N$ to a vector of feasible total utilities to all players.
  $\mu:   \mathbb{R}^N \rightarrow  \mathbb{R}^N$, and
 \beq
 \label{mu-eq}
 \mu(\bm{S},\bm{d},\bm{A}) \overset{def}=\bm{m}=(m_1,\cdots,m_N)
 \eeq
where  $m_n=\frac{1}{N}I_n(\bm{S},\bm{d,A_n})+(1-\frac{1}{N})d_n$ is the midpoint for player $n$.
   \edf
   
 %==============================================================================
%Def:  midpoint lemma
%==============================================================================
   
 \begin{lemma}
Let's assume a bargaining  problem ${\bm{(S, d, u)}}$, with the  ideal points $\{I_n{\bm{(S, d,A_n)}}\}_{n=1}^N$ and the set of allocations, $\{{\bm{A_n}}\}_{n=1}^N$.
Then, the midpoint for player $n$ is uniquely defined by the allocation matrix  $\hat{\bm{A}}=\frac{1}{N}\sum_{p=1}^N \bm{A_p}$,  $\hat{\bm{A}}=\{\hat{\alpha}_{nk}\}$, and is given by
 $ m_n=\sum_{k=1}^K \hat{\alpha}_{nk}u_{nk}$.
 \end{lemma}
 \proof
 The mid point for player $n$ is equal to   $m_n=\frac{1}{N}I_n(\bm{S}_d,\bm{d,A_n})+(1-\frac{1}{N})d_n$,
 \beq
 \begin{array}{rcl}
  m_n&=&\displaystyle\sum_{k=1}^K \hat{\alpha}_{nk}u_{nk}=\frac{1}{N}\sum_{p=1}^N \sum_{k=1}^K \alpha_{nk}^{(p)}u_{nk} \\
          &=&\displaystyle\frac{1}{N}\sum_{k=1}^K \alpha_{nk}^{(n)}u_{nk}+\frac{1}{N}\sum_{p=1, p\neq n}^N \sum_{k=1}^K\alpha_{nk}^{(p)}u_{nk}  \\
          &=&\frac{1}{N}I_n+\frac{N-1}{N}d_n\\
          \end{array}
  \eeq

The following rules are  associated with the bankruptcy  problem and will be applied later to the bargaining solutions.
%=============================================f=================================
%Def: Constrained equal-awards
%=============================================================================
\bdf {Constrained equal-awards  rule-CEA:}  A creditor $n \in \cal{N}$  with a claim $c_n$, will be awarded  $s_n=min\{c_n,\lambda\}=CEA_n(\bm{c},\bm{E})$,  where  $\lambda$ is  chosen such that $\sum{s_n}=E$ if $\sum_{n\in\cal{N}}c_n=E$.
\edf
In other words,  no creditor will be awarded more than  his debt.
%=============================================f=================================
%Def: Constrained equal-losses rule-CES
%=============================================================================

\bdf{Constrained equal-losses rule-CEL:}
 A creditor $n \in \cal{N}$  with a claim $c_n$, will be awarded  $s_n=\max\{o,c_n-\lambda\}=CEL(\bm{c},\bm{E})$,  where  $\lambda$ is  chosen such that $\sum s_n=E$.
\edf

The CES rule  focuses on the losses creditors incur. No creditor will lose more than his debt.

%==============================================================================
%  section 3: Extension to $K$ assets 
%=============================================================================

%==============================================================================
\section{Extension to $K$ assets }
%=============================================================================
\label{sec:section3}

In this section we introduce two bargaining solutions for a multi-asset game under the constraint that the utility of an asset is not transferable. The objective of each players is to get the maximum total sum of the utilities.   The bargaining game is different than the bankruptcy game in several respects. First, each player has a different  utility for each asset. Second, if the bargaining process fails, each player can get some utility (the disagreement point). Therefore the bargaining between the players is only on surplus above what they can get by competition.  
 
 %==============================================================================
\subsection{The Raiffa bargaining solution for a multi-asset Bargaining Game}
%==============================================================================
The Raiffa procedure is a step-by-step  process  where each step increases the utility of all players.  The Raiffa Bargaining Solution (RBS) for a $K$  assets bargaining game is based on iterations of the two rules MP and RIP.  The bargaining is done step by step, where agreement on  the current step becomes the point of disagreement for the next step.  In the initial step the midpoint is set to $\bm{m^{(0)}}=\bm{d}$.
Now,  for each step $j$ we  first apply the RIP rule $N$ times to find  the ideal  allocation  matrices for all players by solving $N$ times the optimization problem in (\ref{RIPeq}).  Then, we apply the MP rule to get the next midpoint vector, $\bm{m}^{(1}$, for the next step. We repeat these steps until the  distance from  the {P}areto frontier is arbitrarily small. The algorithm is shown in Table \ref{N_players_table}. For simplicity of notation, we use $I_n^{(j)}$ to denote  $I_n(\bm{S_d,d,A_n^{(j)}})$.

\begin{lemma}
The above procedure converges to a  $\epsilon$-{P}areto  optimal solution.
\end{lemma}
%==========================================================
\begin{proof}

Assuming that the procedure stops after $j$ steps, and  $ \max_n \vert I_n^j-m_n\vert \leq  \epsilon$, with the final allocation matrix  is $\hat{\bm{A}}=\frac{1}{N}\sum_{p=1}^N \bm{A_p}$. Then the final utility for player $n$ is $m_n$, according to (\ref{RIPeq}). Let $\bm{I_1(\bm{S,d,A_1})}=\{I_1^j, m_2, ...,m_N\}$ be  a point on $\partial \bm{S}$. We have to prove that the point $\bm{m}$ is at a distance of less than $\epsilon$ from the {P}areto frontier. The distance of $\bm{m}$ from the {P}areto frontier is bounded by
\[\| \bm{I_1(\bm{S,d,A})}-\bm{m} \|_1 \leq  \max_n \| I_n^j-m_n \|\leq  \epsilon. \]
 Therefore, the procedure converges to a  $\epsilon$-{P}areto  optimal solution.
\end{proof}

 \textit{Example I}: Assume that we have two players and three assets. The  utility vectors  for player $1$  and $2$ are   $\bm{u_1}=(20,20,30)$, and   $\bm{u_2}=(100, 50,10)$, respectively. The point of disagreement is  $\bm{d}=(0,0)$.  In the first step the ideal points for player $1$ and player $2$ are $70$ and $160$  and  are achieved by the allocation matrices  $\bm{A_1^{(1)}}$ and $\bm{A_2^{(1)}}$, respectively. 
\[ \bm{A_1^{(1)}}=\left(
			\begin{array}{ccc}
				1  & 1  & 1  \\
				 0 &   0&   0 
				\end{array}
\right), 
\bm{A_2^{(1)}}=
\left(
\begin{array}{ccc}
 0 &   0&   0\\
 1  & 1  & 1  
\end{array}
\right).\]
Using the MP rule, the midpoint  of the players is $\bm{m_1}=(35,80)$ with the allocation matrix 
\[\bm{\hat{A}^{(1)}}=
\left(
\begin{array}{ccc}
 .5 & .5&   .5\\
 .5 & .5   & .5  
\end{array}
\right).
\]
In the second step we set the disagreement point to $\bm{d}=\bm{m_1}$ and using the RIP rule to compute the next ideal points for each player.
The ideal points for player $1$ and player $2$ are $54$ and $137.5$, respectively. The  allocation matrices  $\bm{A_1^{(2)}}$ and $\bm{A_2^{(2)}}$  are
\[ \bm{A_1^{(2)}}=\left(
			\begin{array}{ccc}
				0.2  & 1  & 1  \\
				 0.8 &   0&   0 
				\end{array}
\right), 
\bm{A_2^{(2)}}=
\left(
\begin{array}{ccc}
 0 &   0.25&   1\\
 1  & 0.75  & 0  
\end{array}
\right).\]
Applying again the MP rule we obtain that the midpoint  of the players is $\bm{m_1}=(44.5,108.75)$ with the allocation matrix 
\[\bm{\hat{A}^{(2)}}=
\left(
\begin{array}{ccc}
 .1& .625& 1\\
 .9& .375   & 0  
\end{array}
\right).
\]
In the third step we set the disagreement point to $\bm{d}=\bm{m_2}$ and using the RIP rule to compute the next ideal points for each player.
The ideal points for player $1$ and player $2$ are  is $54$ and $137.5$, respectively. The allocation matrices  $\bm{A_1^{(3)}}$ and $\bm{A_2^{(3)}}$ are
\[ \bm{A_1^{(3)}}=\left(
			\begin{array}{ccc}
				0  & .825  & 1  \\
				 1 &   0.175&   0 
				\end{array}
\right), 
\bm{A_2^{(3)}}=
\left(
\begin{array}{ccc}
 0 &   0.725&   1\\
 1  & 0.275  & 0  
\end{array}
\right).\]
The midpoint  of the players is $\bm{m_3}=(45.5,111.25)$ with the allocation matrix 
\[\bm{\hat{A}^{(3)}}=
\left(
\begin{array}{ccc}
0& .775& 1\\
 1& .225   & 0  
\end{array}
\right).
\]

The midpoint  $\bm{m_3}$ is on the {P}areto frontier, and therefore the allocation process ends. The achievable utility  region between the players are inside  set $\bm{S}$ as depicted in Figure  \ref{raiffaN_fig}. The final agreement is reached after three steps. There may be  multiple options in the utility space that provide the same utility in the intermediate steps for sharing the multi-assets.  However, the final allocation is unique.

%==========================================================
% Raiffa Multicomodities alg.
%==========================================================
\begin{table}
\caption{Raiffa bargaining solution procedure  for the multi-asset  Resource Allocation problem \label{N_players_table}}{
\begin{tabular}{||l||}
\hline \hline {\bf Initialization:}  \\
Set $\bm{m^{(0)}}=\left\{d_1, \cdots, d_n\right\}, j=0.$ \\
$\Delta =K, \forall n\in\{0,...,N\}$. \\
$ Set  \ \epsilon=10^{-4}$. \\
\hline \hline {\bm{Computation:}}  \\
While  $\Delta   \geq \epsilon$  \\
\quad Set $j=j+1. $\\
\quad \quad  for n=1:N \\
\quad \quad \quad Find  the ideal point, $ I_n^j $ according to the LP in  (\ref{RIPeq}). \\
\quad \quad \quad Set the initial  midpoint  for player $n$: \\
\quad \quad \quad $m^{j+1}=I_n^j /N+(1-\frac{1}{N})d_n$.\\
\quad \quad \quad $m_n=m_n^j$ \\
\quad \quad end  \\
\quad Set: $\Delta = N \max_n\lVert I_n^j-m_n \rVert_2$ \\
%\quad \quad $\Delta=N\Delta m^j$ \\
End \\
\hline
\hline
\quad  \quad Allocate to player $n$ the assets according to $\alpha$'s, $\{\alpha_{nk}^j\}$, \\
\quad  \quad and the  final utility of player $n$ is $u_n=\sum_{ 1}^K\alpha_{nk}^ju_{nk}=m_n$.\\
End \\
\hline
\hline
\end{tabular}}
\end{table}
%==========================================================
% Main lemma
%==========================================================

\begin{figure}[htbp]
\centering \epsfig{file=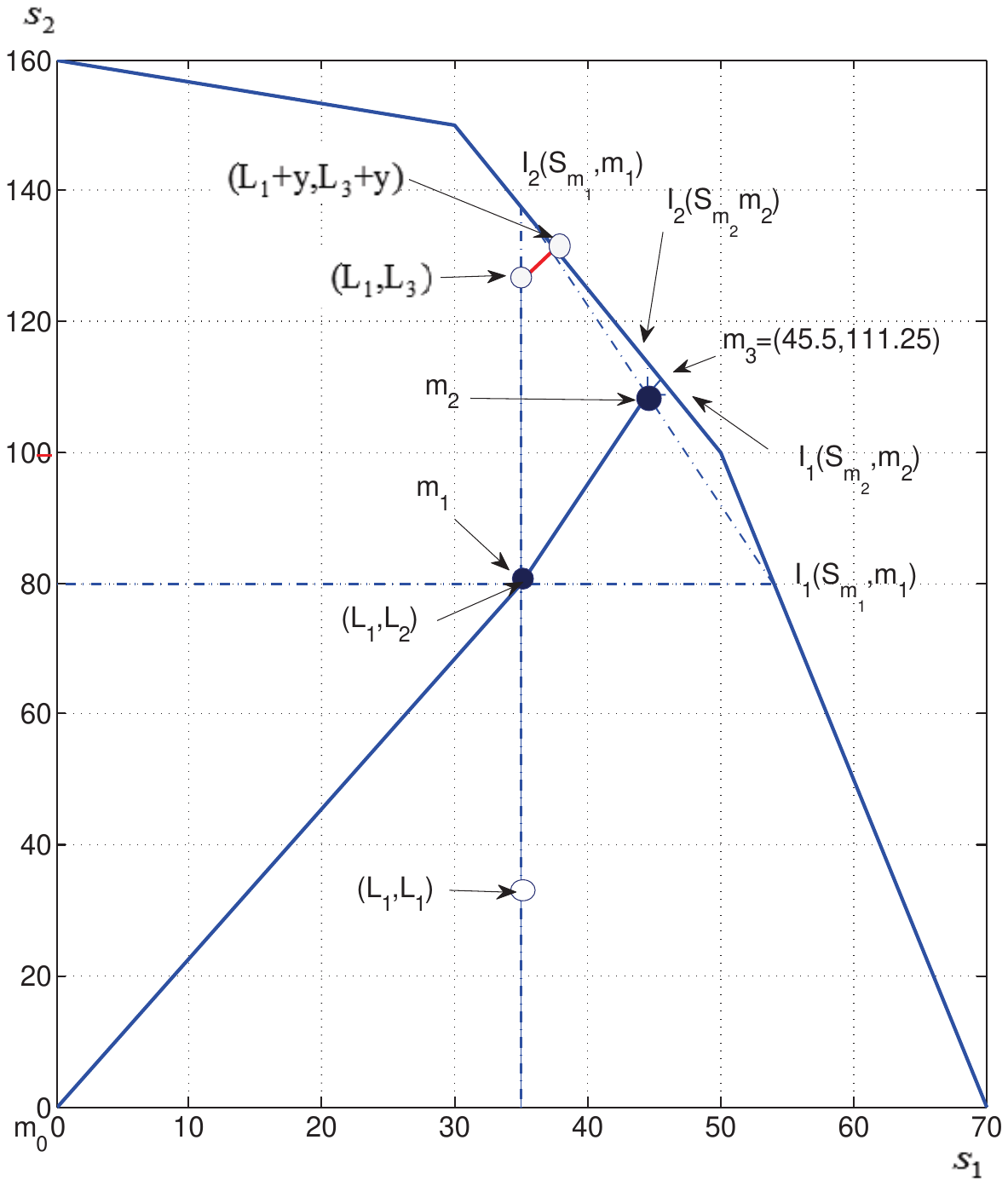,width=0.45\textwidth}
\caption{The Raiffa and Talmud rule solutions for the 2-player bargaining game.}
\label{raiffaN_fig}
\end{figure}

\begin{table}
\caption{Raiffa Solution and TR solution for 3 Players  and 3 Assets \label{3_players_example}}{
\begin{tabular}{||c|c|c|c|c||}
\hline
\hline
\hline
Initial Status	&A 	& B & C		&Ideal Utility \\
\hline
Player 1	&20 	 &20	   &30  	&70\\
\hline
Player 2	&100 &50	  &10		&160\\
\hline
\hline
First Run	&k=1 	& k=2	& k=3		&Ideal Utility \\
\hline
Player 1	&0	&5		&30      &54\\
\hline
Player 2	&80 	&0		&0        &137.5 \\
\hline
\hline
Second Run	&A 	& B & C		&Ideal Utility \\
\hline
Player 1	&0	&14.5		&30      54\\
\hline
Player 2	&100 	&8.75		&0        137.5 \\
\hline
\hline
Third  Run	&A 	& B & C		&Final l Utility \\
\hline
Player 1	&0	&15.5		&30      &45.5\\
\hline
Player 2	&100 	&11.25		&0        &111.25 \\
\hline
\hline
Talmud rule  &A 	& B & C	&Final Utility \\
\hline
Player 1	&0	&8.6	&30	&38.6\\
\hline
Player 2	&100	&28.6	&0	&128.6\\
\hline
\hline
\end{tabular}}
\end{table}
\begin{lemma}
On every step of the RBS the unallocated utility is reduced by a factor of $(1-\frac{1}{N})^N>e^{-1}$.
\end{lemma}
\begin{proof}

Any allocation of assets by a matrix $A$ can be mapped as a point in the  utility space  $ {\bm{S}}$. A new disagreement  point
reduces the distance between the ideal point and the previous disagreement point by a factor of $(1-\frac{1}{N})$, in every coordinate. Therefore, in every step of the RBS the unallocated utility is reduced by a factor of $(1-\frac{1}{N})^N$.
\end{proof}
The Raiffa solution is obtained by solving a series of linear programming problems. The KKT conditions  and the properties of the solution for the case of two players is given in Appendix A. Several comments are in order:

\begin{itemize}
\item A sufficient condition for the existence  solution to the bargaining problem is  that $d_n<\frac{u_n}{N}, \forall N$.  This is because, the allocation of $1/N$  of every  asset  to a player will provide him with a utility greater than what he can get by disagreement. 
\item Set $\bm{S}$ is constructed by a finite number of intersections of hyper-planes. Every intermediate disagreement reduces the number of hyper-planes that define the {P}areto frontiers of the set. Therefore,  the final bargaining result is achieved in finite steps, if  the Raiffa solution is not on the intersection line of two hyper-planes.
\item The linear programming problem for two players can be solved with a complexity of $O(Klog(K))$ (see Appendix A), and the number of assets that are shared by more than one player is at most one.
\end{itemize}

 We now show the relationship between the Raiffa bargaining solution and the CEA rule. Let us  define the modified CEA rule for the  bargaining problem as follows:
\bdf
Let the ideal point of player $n$ be $I_n(\bm{S_d},\bm{d,A})$, and $\partial \bm{S_d}$ is the {P}areto frontier of the  set $\bm{S_d}$; then the modified constraint equal awards of player $n$ is $CEA_n(\frac{I_n(\bm{S_d},\bm{d,A})-\bm{d}}{N}+\bm{d}, \partial \bm{S_d})$
\edf

The Raiffa bargaining solution can be viewed as a repetition of the modified CEA rule. Step $j$ in the Raiffa  bargaining solution can be obtained by applying  the modified CEA rule to the remainder of the asset, namely;  
\beq
 m_{n_{j+1}}=m_{n_j}+CEA_n(\frac{I_n(\bm{S_{m_{j}}},\bm{m_{j},A)})-\bm{m_{j}}}{N}+\bm{m_{j}}, \partial \bm{S_{m_{j}}}).
 \eeq

This interpretation of the Raiffa bargaining solution resembles Piniles' rule, \cite{Piniles61}. Here, we applied the CEA rule N times to take  one step in the Raiffa  bargaining solution.

%==============================================================================
\subsection{Extension of the Talmud rule bargaining solution to the  multi-asset  game}
%==============================================================================

 \cite{aumann_1985} considered the problem of the division of a property  $E$, when the creditors have debts $c_1, ..., c_n$, that  are worth more than  $E$.  They proposed  allocating the  property according to an extension of the  Contested Garment case , which is also known as the Concede-and-Divide rule, as follows:

\bdf{Concede-and-Divide, (CD) Rule for two players:}

Two creditors have claims $c_1$ and $c_2$ on a  property   $E$. The amount that  creditor $i$ will be awarded is
\[s_i=\frac{E- (E - c_1)_{+}-(E - c_2)_{+}}{2}+(E - c_{3-i})_{+}, i=1,2, \] where, $(x)_{+}=\max(x,0)$. We denote this division as $CD(c_1,c_2,E)=(s_1,s_2).$
\edf
Any two creditors $i$ and $j$ will be awarded $s_i$ and $s_{j}$, such that  $CD(c_i,c_j,E_{i,j})=(s_i,s_j), E_{i,j}=s_i+s_j, \forall i, j$ and $\sum_{n=1}^Ns_n=E$.
Assume that the debts are ranked in increasing order: $c_1\leq,..., \leq c_n$, and $C=\sum_{n=1}^N c_n$. Aumann  proposed the following allocation called the Talmud rule {(TR)}: each creditor  will get  $s_n$, where
\beq
s_n=\left\{
\bea{ll}
\min\left\{ \frac{c_n}{2},\lambda\right\} =CEA_n(\bm{c}/2,E))&   E \leq \frac{1}{2}C\\
\max\left\{ \frac{c_n}{2},c_n -\mu \right\}=CEL_n(\bm{c},E)& E > \frac{1}{2}C \\,
 \ena \right.
 \eeq
 and $\lambda$ and $\mu$ are chosen to satisfy the constraint $\sum_{n\in\cal{N}}c_n=E$.

 %==============================================================================
%\Start modify to multu aset 
%==============================================================================

The Talmud\footnote{The rules in the TAL-family \cite{moreno2006} extends  the Talumd  rule by using a parameter of $\theta\in\left[0,1\right]$; i.e., nobody gets more than a fraction $\theta$ of his claim if the asset value to divide is less than $\theta$  times of the aggregate claim,  and nobody gets less than a fraction  $\theta$ of his claim if the asset value to divide is  larger than  $\theta$  times the aggregate claim.
The extension of the Talmud rule to the multi-asset bargaining problem can be easily extended to the TAL-family rules.} rule can be interpreted  as a composition of the constrained equal awards rule and the constrained equal losses rule. In the Talmud rule no creditor gets more than half of his claim if the asset value is less than half of the aggregate claim and nobody gets less than half of his claim if the asset value  exceeds half of the aggregate claim.
 
We now modify the bankruptcy solution and apply it to the bargaining problem. Let us adopt the following modification:
\begin{itemize}
\item  The  total utility  that  a player $n$  claims is  $I_n(\bm{S}_d,\bm{d},\bm{A_n})$,  and without an agreement he gets $d_n$. Therefore,  the negotiation is only on the  surplus $c_n=I_n(\bm{S}_d,\bm{A})-d_n$.  For simplicity of notation, we use $I_n$ to denote  $I_n(\bm{S_d,d,A_n})$, and $D$ for $\sum_{n=1}^Nd_n$.
\item The value of $I_n$ is obtained by  solving  the optimization problem in (\ref{RIPeq}). 
\item  In contrast to the bankruptcy problem, there is no single asset with a value of $E$ that has to be divided between the players. Here, each player claims the all assets. The solution will be on the {P}areto frontier of set $\bm{S}$.  Therefore, the CD rule for the bargaining problem between players $i$ and $j$ is  $CD(c_i, c_j, E_{i,j})=(s_i, s_j)$,  $E_{i,j}=s_i+s_j$, and  $\bm{s}  \in {\partial \bm{S}}$.
\end{itemize}

The extension of the Talmud rule solution to the multi-asset case  is based on a binary search of a {P}areto optimal allocation that satisfies the CD rule  and the RIP rule.   The CD rule defines $2N+1$ levels, $L_n$ (see Figure \ref{aumann1_fig}), where each level corresponds to  a point in $\mathbb{R}^N$, that can be either inside the set $\bm{V}\subset \mathbb{R}^N $ or outside the set. The bargaining solution has to be on the {P}areto frontier of the total utility space, and defines a unique water level $L$. The 
\[ L_n=\left(
			\begin{array}{cc}
				0 & n=0\\
				\frac{c_n}{2}  & 1 \leq p\leq  N   \\
				\end{array}
\right)
;  L_{2N-p}=\left(
			\begin{array}{cc}		
				c_N+d_N-\frac{c_n}{2}  & 1\leq p <N   \\
				c_N+d_N    & p=0
				\end{array}
\right) \]

%The modified algorithm is presented in Table \ref{Aumann1}. } 
We  use the  \cite{kaminski_2000} water-filling  interpretation to describe the algorithm.  Figure \ref{aumann1_fig}  depicts $N$ containers of differing sizes, representing the claims of the players, into every one of which we pour water representing the utility. A  container representing the claim of a player $n$  is divided into two halves connected by a narrow tube that allows the water to run through it, but  with almost zero capacity. All the containers are connected at level $L_0=0$  by a tube that likewise is very narrow but allows the fluid to pass between containers according to the law of communicating containers. All containers are at the same height above the ground, $L_0$, and width and have a different tube. The container with the smallest capacity has the longest tube. The lower part of the container has a capacity that is equal to half of the claim of  player-$L_n=c_n/2, n \leq N$ above level $L_0$,  plus what he can get by competition-$d_n$. The upper part  has a capacity equal to half of the claim of the player.  Thus, the capacity of a container represents the player's claim plus $d_n$.  The containers are ranked according to their claims.

We now pour water into  all the containers. If the extra  utility to be shared is between $D$ and $D+N\cdot L_1$, all the containers (players)  share the water (extra utility) equally, and the water level in all the containers is  the same according to the law of communicating containers.  If the extra utility is greater than $D+N\cdot L_1$,  the container (player) with the smallest (volume) claim stops receiving anything for a while, and the water  is divided equally among all the other  containers until each container has an amount equal to the second smallest half-claim $L_2 $ plus $d_n$.  This process continues as follows:  whenever the water level is above $L_p$ player $p$ stops receiving anything, while the rest of the players share the water (extra utility) equally. 
Therefore, whenever the extra utility to be shared by the players is smaller than  the half-sum of the
claims plus $D$; i.e., $\sum_{n=1}^N L_n+d_n$,  each player receives at most his half claim according  to the constrained equal awards rule.

When the extra utility  exceeds half the sum of the claims, the calculation is made in accordance with each player's  losses:  the difference between the player's claim $\frac{c_n}{2}$ and what he actually  gets is $b_n$. Now, if the water level is between $L_{2N-p+1}$, and $L_{2N-p},  1\leq p<N$, the  water is shared equally between the $p$'th  container and $N$'th  container   according  to constrained equal losses rule.

The algorithm consists of several steps (see Table \ref{Aumann2}). In the first step we need to find what rule to apply: the Constraint Equal-Awards (CEA) rule or  the Constraint Equal-Losses  (CEL) rule. This can be resolved by determining whether there is a feasible allocation if the water level is above $L_N$. If so,  the CEL rule  is applied; otherwise the CEA rule.  The decision is made by solving the  following  linear programming problem:
\begin{equation}
\label{eq2}
\begin{array}{rl}
\textbf{min}       &\sum_{n=1}^N\sum_{k=1}^K\alpha_{nk}\\
\textbf{subject to:} & \forall k \sum_{n=1}^{N}\alpha_ {nk} \leq 1,\quad \\
                     & \forall n,k \ \alpha_ {nk} \geq0, \\
                     & \forall n \sum_{k=1}^K\alpha_{nk}u_{nk}= \frac{c_n}{2}+d_n
\end{array}.
\end{equation}
If there is a result  (the result is in set $\bm{S^{NK}}$),  this implies that the water level $L$ is above $L_N$ and the allocation is according  to the CEL rule; otherwise the water level is below $L_N$ and the allocation is according to the CEA rule. We now explore these two cases.
%==========================================================================================
% CEA Rule
%= ===========================================================================================
\textbf{Case A: The CEA rule}

All the players can gain at most  $\frac{c_n}{2}+d_n$.  Let $p$ the larger number  such that
$\left\{ a_1,...,a_N\right\} \in \bm{S^{N}}$, and
\beq
\label{Aterms}
a_n= \left\{
 \bea{ll}
\frac{c_n}{2}+d_n &  n \leq p \\
\frac{c_p}{2}+d_p & p< n \leq N\\
 \ena \right. .  \\
 \eeq
This problem can be formulated as the following linear programming problem;
\begin{equation}
\label{Aequation}
\begin{array}{rl}
\textbf{min}       &\sum_{n=1}^N\sum_{k=1}^K\alpha_{nk}\\
\textbf{subject to:} &\sum_{n=1}^{N}\alpha_ {nk} \leq 1, \ \  \forall k,\quad \\
                     &  \alpha_ {nk} \geq0, \ \  \forall n,k \\
                     &  \sum_{k=1}^K\alpha_{nk}u_{nk}= a_n, \ \forall n
                     \end{array}.
\end{equation}
Here, $p$ can be found by a binary search. The exact water level $L$  has to be  above $L_{p}+y$, but below the next level $L_{p+1}$. All players  with an index greater than $p$ will share  the extra utility  equally, and $y$ is the result of the following  linear programming problem
\begin{equation}
\label{Bequation}
\begin{array}{rl}
\textbf{max}    &  y \\
\textbf{subject to:} &  \sum_{n=1}^{N}\alpha_ {nk} = 1,\ \forall k,\quad \\
                     &  \alpha_ {nk} \geq0, \ \ \forall n,k, \\
                     &  \sum_{k=1}^K\alpha_{nk}u_{nk}= b_n, \forall n
                     \end{array},
\end{equation}
and $b_n$ is given by  
\beq
\label{Bterms}
b_n= \left\{
 \bea{ll}
 d_n +y&    p=0, \forall n \\
\frac{c_n}{2} +d_n&  1\leq n\leq p\\
\frac{c_p}{2}+d_n +y&  p< n \leq N\\
 \ena \right. .\eeq \\
The allocation to  player $n$ of asset  $k$ is  $\alpha_{nk}$, where $\left\{\alpha_{nk}\right\}$ is the solution to equation (\ref{Bequation}).
%==========================================================================================
% CEL Rule
%= ===========================================================================================
\textbf{Case B: The CEL rule}

All players lose at most  $\frac{c_n}{2}$. Let $p$ be  the smallest  $p$ such that  $\left\{ a_1,...,a_N\right\} \in \bm{S^{N}}$, where $a_n$ is given by  
\beq
\label{Aterms1}
a_n= \left\{
 \bea{ll}
\frac{c_n}{2}+ d_n &  n \leq p \\
c_N+d_N-\frac{c_p}{2} &   p<n \leq N\\
 \ena \right. .
 \eeq
Similar to  (\ref{Aequation}) with different values for the $a_n$'s, $p$ can be found by a binary search.

The exact water level has to be  above $L_{2N-p-1}$, but below the next level $L_{2N-p}$. All players  with an index  equal or greater than $p$ will share  the extra utility equally, and  $y$ is the result  to the linear programming problem in equation ({\ref{Bequation}}), where  $b_n$ is given by 
\beq
\label{Bterms1}
b_n= \left\{
 \bea{ll}
\frac{c_n}{2}+d_n &   n \leq p \\
c_N+d_N-\frac{c_p}{2} +y &  p <q n \leq N
 \ena \right. .
 \eeq
The allocation to  player $n$ of asset  $k$ is  $\alpha_{nk}$, where $\alpha_{nk}$ is the solution to equation (\ref{Bequation}).

%-------------------------------------------------------------------------------------------------------
% Table for multi asset 
%------------------------------------------------------------------------------------------------------
\begin{table}
\caption{Talmud rule  bargaining solution for the multi-asset  case\label{Aumann2}}
{\begin{tabular}{||l||}
\hline \hline {\bf Initialization:}
Solve the linear programming  problem in \\
equation  (\ref{eq2}). If there is a solution  then go to B, \\
otherwise go to A. \\
\hline
\hline
A. All players gain at most  half of $\frac{c_n}{2}+d_n.$\\
1. Do a binary search to find the  larger   $p$ such that there \\
is a solution to the linear programming in (\ref{Aequation}), \\
and   $a_n$ are given in (\ref{Aterms})\\
2. Solve the linear programming in (\ref{Bequation}), \\
and   $b_n$ are given in (\ref{Bterms})\\
3. $b_n$ is the total utility that is allocated to player $n$.\\
The assets are allocated according to the allocation matrix $A$. \\
Exit\\
\hline
\hline
B. All players lose at most half of the debt.\\
1. Do a binary search to find the smallest  $p$ such that there \\
is a solution to the linear programming in (\ref{Aequation}), \\
and   $a_n$ are given in (\ref{Aterms1})\\
2. Solve the linear programming in (\ref{Bequation}), \\
and   $b_n$ are given in (\ref{Bterms1})\\
3. $b_n$ is the utility that is allocated to player $n$.\\
The assets are allocated according to the allocation matrix $\bm{A}$.\\
\hline \hline
\end{tabular}}
\end{table}
%===========================================================
% Figure water
%===========================================================
\begin{figure}[htbp]
\centering \epsfig{file=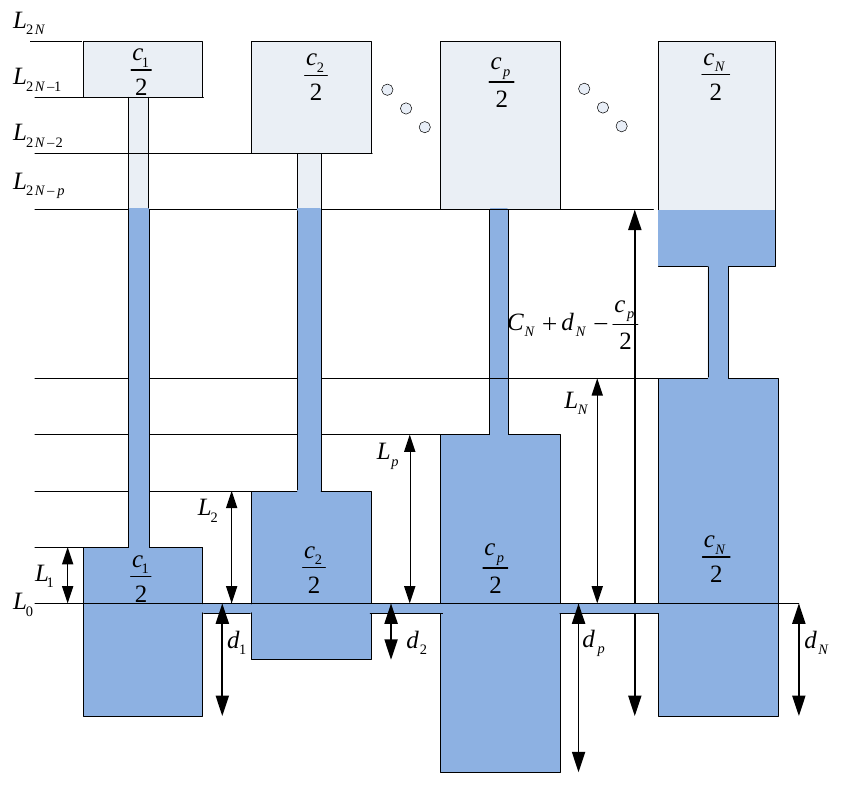,width=0.70 \textwidth}
\caption{Water-filling interpretation of the Talmud rule for bargaining solution.}
\label{aumann1_fig}
\end{figure}

 \textit{Example II}: Assume the same conditions as in Example I . The ideal point for player $1$ and player $2$ is $70$ and $160$, respectively, and both claim to get the utility of the ideal point. An allocation of half of the claim for each player is  also inside the set, $(L_1,L_2)=(35,80) \in\bm{V}$. The point $(L_1,L_3)=(35,125)$ is also inside $\bm{V}$, so the solution is to share the remaining utility equally between the players. Hence, we have to search for a point  $(u_1,u_2)=(35+y,125+y)$ that is on the {P}areto frontier of set $\bm{V}$, namely $s_2=150-\frac{5}{2}(s_1-30).$ Thus, point $(L_1+y,L_3+y)=(38\frac{4}{7},128\frac{4}{7})$ provides equal losses to both players.  Set $\bm{V}$ and the solution of the bargaining process are depicted in Table  \ref{3_players_example} and in Figure \ref{raiffaN_fig}. The allocation matrix in this case is 
\[\bm{A}=
\left(
\begin{array}{ccc}
0& \frac{3}{7}& 1\\
 1& \frac{4}{7} & 0  
\end{array}
\right).
\]

\newpage
\section{ Properties and complexity}
\label{sec:section4}

\subsection{ Properties of Raiffa and TR bargaining solutions}
In  this subsection we define some of the properties that relate to the bargaining solutions that described in the previous section. Here we adopt the 
defintions in \cite{thomson_2003} and  \cite{hinojosa2012}.

Let $\Pi_N$ denote the class of bijections from $\cal{N}$ to itself,  $\Pi^K$ denotes the class of bijections from $\cal{K}$ to itself. Let $\pi \in \Pi_ N$, and denote by $\bm{S_\pi}$ the matrix whose $k$th row is $s_{\pi(k)}$ for $n\in \cal{N}$. We also define  $\sigma \in \Pi_K$, and denote by $\bm{S^\sigma}$ the  $k$th column after column permutation. \\

\emph{Anonymity}: For each $\Pi_N$, and each $n\in \cal{N}$, if  $\varphi (\bm{ S,d,u})=\bm{A}$, then   $\varphi_{\pi}(\bm{ S_{\pi},d_{\pi},u_{\pi}})=\bm{A_\pi}$.

\emph{Neutrality}: For each $\Pi^K$, and each $k\in \cal{K}$, if  $\varphi (\bm{ S,d,u})=\bm{A}$, then   $\varphi^{\sigma}(\bm{ S^{\sigma},d^{\sigma},u^{\sigma}})=\bm{A^\sigma}$.

Anonymity and neutrality hold for both solutions. \\

\emph{Equal treatment of equal}: For a bargaining game $(\bm{S},\bm{d},\bm{u})$, and players $n,p\in \cal{N}$, if $I_n-d_n=I_p-d_p$, then $\varphi_n(\bm{S,d,u})\bm{u_n}^T-d_n=\varphi_p(\bm{S,d,u})\bm{u_p}^T-d_p$.\\

This property states that any two players that have the same claims ($c_n=I_n-d_n$) will get the same award. \\

\emph{Order preservation}: For a bargaining game $(\bm{S},\bm{d},\bm{u})$, and players $n,p\in \cal{N}$, if $I_n-d_n\geq I_p-d_p$, then $\varphi_n(\bm{S,d,u})\bm{u_n}^T-d_n \geq \varphi_p(\bm{S,d,u})\bm{u_p}^T-d_p$, and  $I_n-\varphi_n(\bm{S,d,u})\bm{u_n}^T \geq  I_p -\varphi_p(\bm{S,d,u})\bm{u_p}^T$\\

 \begin{lemma} 
 The TR bargaining solution satisfies equal treatment of equal and order preservation. 
 \end{lemma}
 \proof: 
 
The claims of all players  are calculated  once, according to the {RIP} rule at the initial phase of the algorithm. The water filling algorithm ensures  that the player with the smaller claim will gain (and lose)  less than the player with the larger claim, and  any two players that have the same claims  will get the same award based on the same argument. In the case of the Raiffa bargaining solution there is no guarantee that these properties will hold in each phase of the algorithm.

\emph{Homogeneity}. For  any  bargaining game $(\bm{S},\bm{d},\bm{u})$ and $\lambda>0$, $\varphi(\bm{\lambda S,\lambda d,\lambda u})=\bm{A}$.

\begin{lemma} 
 The TR bargaining solution and the Raiffa bargaining solution  satisfy homogeneity. 
 \end{lemma}
 
  \proof: 
 These bargaining solutions are based on solving the LP problems. In these problems,  $\lambda$ multiplies  both sides of the  constraints,  and therefore   will not change the solution of the optimization problems.\\
 
  Any point on the {P}areto frontier of set $\bm{S}$ can be obtained by assigning a proper weight vector  $\{w_1,\cdots,w_n\}$, and solving the corresponding weighted max-min optimization problem. \cite{zehavi_2013} proved that for a weighted max-min allocation problem  of $K$ assets to
$N$ players, there is always a result where  at most $N-1$ assets are  shared by more than one player. The Raiffa bargaining solution and the Talmud rule solution are {P}areto optimal solutions located on the {P}areto frontier. Therefore, in these solutions,  the number of assets that are shared by more than one player is at most $N-1$. Note  that if the number of assets is very large in comparison to the number of players, it is easy
to modify the allocation such that each player loses at most a single asset  it shares with others, and that this loss is small when $N>>K$.

Note that in the case of 2 players the Talmud rule solution always operates according to the constrained equal losses rule (due to the convexity of set $\bm{S}$). It is easy to show that the player with the larger claim gets more than in the Raiffa bargaining solution.  However when the number of players is greater than two, and the players are bargaining on a single asset, the CEA rule applies.
%==================================================
% new subsection
%------------------------------------------------------------------------------------------------
\subsection{ Complexity of the Raiffa and TR bargaining solutions}
Since the discovery of the Simplex Method in the 1940s, extensive work has been done on algorithms for solving  {Linear Programming} ({LP}).
Large numbers of optimization algorithms have been developed including variants on the {Simplex Method}, the {Ellipsoid Method},  and the  {Primal-Dual Interior-Point Method}.   \cite{nem} proved in 1979 that Linear Programming is polynomially solvable; namely, that an LP problem
with rational coefficients, $m$ inequality constraints and $n$ variables can be solved in $O(n^3(n+m)L)$
arithmetic operations, where $L $\footnote{ In our case  K(N-1)  coefficients of matrix A are zeros or ones;  thus, the value of $L$ is bounded by $O( K^2N)$, since $ L =\sum_{i,j} \log_2(a_{i,j}+1)+\log_2(nm)+(nm+m)=O( K^2N)$}
is the input
length of the problem; i.e., the total binary length of the numerical data specifying the problem instance. In our case (the primal dual problem) we have $n=K N+1$ variables and $m=K+N$ inequality constraints.  Note that the matrix in our case is almost unimodular, and sparse; thus the worst case complexity is on the order of $O(K^6N^5)$. Hence, the complexity of RBS is $O(JK^6N^6)$, where $J$ is the number of iterations, and the complexity of  the Talmud rule solution is $O(K^6N^5)$. In practice, the algorithms  converge faster than the worst case bound. A more extensive discussion on complexity can be found in  \cite{nem}.

\subsection{ Examples}
The  utility of an asset is not transferable;  therefore, the allocation matrix  depends highly  on the utility of each asset to the player. Table \ref{1case} presents a scenario with three players and  seven assets. Here, the utility of the assets for each player  is given in rows 3-5 of the table, and the ideal points of the players are $(24.2, 83.3, 103.2)$, respectively. The allocations of the assets  for each player  according to RBS  are in rows 7-9, and the allocations  according to TR are in rows 11-13. The final allocation for each player according to the RBS and TR are  given in rows 15-17. Similar  results  are shown for a different scenario in  Table \ref{2case}, where the ideal points  of the players are $(21.2, 51.4, 74)$, respectively. However, the allocations for each player in scenario II are higher than in scenario I, due to the fact that most of the players, utility  is concentrated in different assets (scenario II). In the case of more than two players, allocating at least half of the ideal points to all players is sometimes not feasible (scenario I). In this case,  the Talmud rule solution either allocates half of the claim to the player  with the weakest ideal point or  allocates the utility equally among all the players.

\begin{table}
\caption{Scenario I: Three players and  seven assets \label{1case}}{
\begin{tabular}{|c|c|c|c|c|c|c|c|}
\hline
\hline
Scenario I	& \multicolumn{7}{ |c| } { Utility of assets}\\
\hline
Player 	& 1&  2 &  3 & 4&  5 &  6& 7 \\
\hline
1	& 3.0    &4.7  &  2.3   & 8.4    & 1.9    & 2.2  &   1.7  \\
 \hline
 2  &8.7   & 6.2   & 18.4  &  8.6   &  3.7  & 18.1  & 19.6  \\
 \hline
 3 & 3.9    &9.0   &14.3  &20.8   & 9.2  & 21.1 &  24.9  \\
\hline
Player 	  	  & \multicolumn{7}{ |c| } {Raiffa- Commodities allocation}\\
\hline
 1   & 0   & 1   &  0 &   0.844   &  0 &   0&   0  \\
\hline
 2    & 1 &  0 &   1   & 0  &   0  &  0.58   &  0  \\
\hline
 3    &0  & 0   & 0   &  0.156  &  1  &   0.42   &  1\\
 \hline
Player 	   	  & \multicolumn{7}{ |c| } {Talmud rule  Commodities allocation}\\
	   \hline
1   &  0 & 1   &0&    0.88&   0&   0&   0  \\
 \hline
 2    & 1 &  0   &1& 0& 0&  0.78& 0 \\
 \hline
3  & 0       &0  &0 &0.12&    1& 0.22&    1 \\
  \hline
   \hline
       	  & \multicolumn{7}{ |c| } {Sum of utilities per player }\\
Player 	   	  & \multicolumn{3}{ |c| } {Raiffa }& \multicolumn{4}{ |c| } {Aumman }\\
	   \hline
 1   & \multicolumn{3}{ |c| } {11.7930  }& \multicolumn{4}{ |c| } {12.100}\\
 \hline
 2    & \multicolumn{3}{ |c| } {37.5995}& \multicolumn{4}{ |c| } {   41.218 }\\
 \hline
 3  & \multicolumn{3}{ |c| } {46.1966 }& \multicolumn{4}{ |c| } {   41.218}\\
\hline
\hline
\end{tabular}}
\end{table}

\begin{table}
\caption{Scenario II: Three players and  seven assets \label{2case}}{
\begin{tabular}{|c|c|c|c|c|c|c|c|}
\hline
\hline
Scenario II & \multicolumn{7}{ |c| } { Utility of assets}\\
\hline
Player 	& 1&  2 &  3 & 4&  5 &  6& 7 \\
\hline
1&8.4 &   8.7   & 3.0   & 0.1  &  0.2  &   0.5   &  0.3 \\
\hline
2 &    0.3  &  0.2  & 18.5    &12.1   &19.6  &  0.5  &  0.2 \\
\hline
 3&    0.2 &    0.7  & 10.5  &  0.1   & 1.0  & 31.1  & 30.4 \\
  \hline
Player 	    	  & \multicolumn{7}{ |c| } {Raiffa- Commodities allocation}\\
\hline
1&    1&    1&  0.126&    0&  0&   0&    0\\
    \hline
2&    0&   0&    0.792&    1&  1&   0&   0\\
    \hline
3&    0&    0&   0.082&    0&    0&    1&   1\\
 \hline
Player 	   	  & \multicolumn{7}{ |c| } {Talmud rule- Commodities  allocation}\\
  \hline	
1  &    1&    0.529&   0&    0&     0&    0&    0 \\
  \hline
 2&       0&    0&   0.622&    1&     1&  0&   0\\
   \hline
3&    0&    0.471&   0.378&     0&    0&     1&   1\\
  \hline
   \hline
       	  & \multicolumn{7}{ |c| } {Sum of utilities per player }\\
Player 	   	  & \multicolumn{3}{ |c| } {Raiffa }& \multicolumn{4}{ |c| } {Aumman }\\
	   \hline
 1   & \multicolumn{3}{ |c| } {17.4770  }& \multicolumn{4}{ |c| } {13.0017}\\
 \hline
 2    & \multicolumn{3}{ |c| } {46.3581} & \multicolumn{4}{ |c| } {43.2017}\\
 \hline
 3  & \multicolumn{3}{ |c| } {62.3611 }& \multicolumn{4}{ |c| } {65.8017}\\
 \hline
\hline
\end{tabular}}
\end{table}

\section{Summary and Conclusion}
\label{sec:conclusions}
The goal of this paper was to extend the  Raiffa bargaining solution and the Talmud rule  to  a multi-asset  game with N-players. We address a bargaining game problem where the utility of each asset is not transferable. The bargaining between the players is on the surplus utility that they can get above  the total utility of the disagreement point. We propose a bargaining model where the players' surplus utilities ($I_n-d_d$)  are  in  lexicographic
 order. This work  modifies a previous model where the claims were ordered according to the minimum utility of the  asset of each player. 
We show that global bargaining solutions can be obtained by solving a sequence of linear programming problems.  The complexity of the solution for RBS with $M$ assets is equivalent to solving $N$ linear programming problems in each step, and the unallocated utilities decrease by a factor of $e$ in  each step. The complexity of the solution for the TR solution  is equivalent to solving $log(2N)$ linear programming problems.  

\section{Appendix}
\begin{appendix}
\section{The linear programming  solution for 2 players}
For the two player case the linear programming
problem  can be dramatically simplified, and we
provide an $O(K \log_2 K)$ complexity algorithm ($K$ is the number of assets). We show that the two players share
at most a single commditiy, regardless of the ratio between the users.
To that end let,
$\alpha_{1k}=\alpha_k$, and
$\alpha_{2k}=1-\alpha_k$.

We want to solve the following optimization problem:

\beq
\begin{array}{rcl}
L(\vga, \vgd, \bm\mu, \vgl)&= &-\displaystyle \sum_{k=1}^K(1-\alpha_{k})u_{2k}-\sum_{k=1}^K\mu_{k}\alpha_{k}\\
&&+\lambda\left(\sum_{k=1}^K\alpha_{k}u_{1k}-u_1\right) .
\end{array}
\eeq

To better understand the problem, we first derive the KKT conditions \cite{boyd04}.
Taking the derivative with respect to  $\alpha_{k}$, we obtain
 \beq
\label{KKT1}
 \begin{array}{c}
u_{2k}+\lambda u_{1k}- \mu_{k}=0.\
\end{array}
\eeq
with the complementarity conditions:
\beq
 \begin{array}{cl}
 \label{KKT2}
1.& \sum_{k=1}^K\alpha_{1k} u_{1k}=u_1,\\
2.&  \mu_{k}\alpha_{k}=0, \ \mu \geq0.  
\end{array}.
\eeq
Based on (\ref{KKT1})-(\ref{KKT2}), we can easily see that the Lagrange multipliers in  (\ref{KKT2}) satisfy the following conclusions :
\begin{itemize}
\item[1.] $\mu_{nk}= 0$,  if $\alpha_k>0, \ \forall\ k$ (see (\ref{KKT2}.2)).
\item[2.]   If $0<\alpha_{1k}^j<1$, then the players share an asset  $p$ if $\frac{u_{2p}}{u_{1p}}=-\lambda$ (see (\ref{KKT1}.2)).
\item[3.]   Asset $p$ is assigned to player $2$ if  $\frac{u_{2p}}{u_{1p}}>-\lambda$.
\item[4.]   Asset $p$ is assigned to player $1$ if  $\frac{u_{2p}}{u_{1p}}<-\lambda$.
\item[5.]   $\sum\alpha_{1k}u_{1k}=u_1$.
\end{itemize}
Assuming that a feasible solution exists and that the assets are  ranked in decreasing order according to the ratio
$L\left(k\right)=\frac{u_{1k}}{u_{2k}}$, it follows from  the KKT conditions that the allocation is made according to
the following rules
\begin{enumerate}
\item  The ideal point of player 1 is $I_1$ given by
\beq
\begin{array}{lcl}
I_1(u_2)&=&\sum_{k=1}^{p-1}u_{1k}+\alpha_pu_{1p},
  \end{array}
\eeq where $p$ and $\alpha_p$ are set such
that \beq \label{kmaxd}
u_{2}=\sum_{k=p}^{K}u_{2k}-\alpha_p u_{2p}.
\eeq
\item  Similarly, the ideal point of player 2 is $I_2(u_1)$ is given by
\beq
\begin{array}{lcl}
 I_2(u_1)& =&\sum_{k=p}^{K}u_{2k}-\alpha_pu_{2p},
\end{array}
\eeq where $p$ and $\alpha_p$ are set such
that \beq \label{kmind}
u_{1}=\sum_{k=1}^{p-1}u_{1k}+\alpha_pu_{1p}.
\eeq
\end{enumerate}
Therefore, no more than one asset  can  be shared by the two players. The algorithm for computing the ideal point of player $1$ is as follows.  Let $L_k=\frac{u_{1k}}{u_{2k}}$ be
the ratio between the utilities of asset  $k$. We can sort the
assets in decreasing  order according to $L_K$.  If all
the values of $L_{k}$ are distinct, there is at most a  single
asset  that has to be shared between the two players. Since
only one  asset   satisfies equation (\ref{kmind}), we denote this
asset $k_s$, and  all the assets $1\leq k<k_s$
will  be allocated to player $1$, while all the assets $k_s <
k \leq K$ will be be allocated to player $2$. Asset  $k_s$ must
be shared accordingly between the players. The  complexity of this algorithm is at most $O(K\log K)$, due to the sorting operation. For the Raiffa bargaining solution  the sorting  operation only has to be done once at the beginning. The complexity of computing the next disagreement point is on the order of $O(K)$.
\end{appendix}

%===========================================================

\bibliographystyle{elsarticle-num}

\end{document}